\documentclass[a4paper,onecolumn,11pt,accepted=2024-06-26]{quantumarticle}
\pdfoutput=1
\usepackage[utf8]{inputenc}
\usepackage[english]{babel}
\usepackage[T1]{fontenc}
\usepackage{amsmath}
\usepackage{hyperref}
\usepackage[numbers,sort&compress]{natbib}
\usepackage{tikz}
\usepackage{lipsum}
\usepackage{xspace}%
\usepackage{multirow}

\usepackage{subcaption}

\usepackage{amsmath}
\usepackage{lipsum}
\usepackage{amsfonts}
\usepackage{amsthm}
\usepackage{graphicx}
\usepackage{epstopdf}
\usepackage{algorithmic}
\usepackage{tikz}
\usetikzlibrary{calc}
\usepackage{cleveref}

\newtheorem{theorem}{Theorem}
\newtheorem{observation}[theorem]{Observation}
\newtheorem{definition}[theorem]{Definition}
\newtheorem{conjecture}[theorem]{Conjecture}
\newtheorem{lemma}[theorem]{Lemma}
\usepackage{float}
\usepackage{leftindex}
\ifpdf
  \DeclareGraphicsExtensions{.eps,.pdf,.png,.jpg}
\else
  \DeclareGraphicsExtensions{.eps}
\fi

\begin{document}

\title{Graph-theoretic insights on the constructability of complex entangled states}

\author{L. Sunil Chandran}
\affiliation{Indian Institute of Science, Bengaluru.}
\orcid{0000-0001-5451-6975}
\thanks{ Supported by SERB Core Research Grant CRG/2022/006770: ``Bridging Quantum Physics with Theoretical Computer Science and Graph Theory''}
\author{Rishikesh Gajjala}
\email{rishikeshg@iisc.ac.in}
\orcid{0000-0002-8726-3465}
\affiliation{Indian Institute of Science, Bengaluru.}
\maketitle

\begin{abstract}
The most efficient automated way to construct a large class of quantum photonic experiments is via abstract representation of graphs with certain properties. While new directions were explored using Artificial intelligence and SAT solvers to find such graphs, it becomes computationally infeasible to do so as the size of the graph increases. So, we take an analytical approach and introduce the technique of local sparsification on experiment graphs, using which we answer a crucial open question in experimental quantum optics, namely whether certain complex entangled quantum states can be constructed. This provides us with more insights into quantum resource theory, the limitation of specific quantum photonic systems and initiates the use of graph-theoretic techniques for designing quantum physics experiments.
\end{abstract}

\section{Introduction}
Recent years have seen dramatic advances in quantum optical technology \cite{intro1, intro2} as photons are at the core of many quantum technologies and in the experimental investigation of fundamental questions about our universe’s local and realistic nature. Due to the peculiar behaviour of multi-particle interference, designing experimental setups to generate multi-partite entanglement in photonics is challenging. The most efficient automated way to construct a large class of such quantum photonic experiments is via abstract representation of graphs with certain properties \cite{Quantum_graphs_3, Quantum_graphs, Quantum_graphs_2} allowing new possibilities for quantum technology \cite{intro3, Feng:23, Bao23}. The construction of such graphs has been a challenging mathematical open problem as one needs to carefully tune the edge weights to satisfy exponentially many equations \cite{krenn2019questions}.

Recently, new directions were explored using Artificial intelligence and SAT solvers to find such graphs, which could be used to design quantum photonic experiments \cite{AIquantum1, AIquantum2, CerveraLierta2022designofquantum}. However, using these methods, it is computationally infeasible to find solutions in large graphs as the search space grows exponentially. Therefore, more advanced analytical methods are necessary. In this work, we introduce the technique of local sparsification on the graphs corresponding to experiments, using which we answer a crucial open question in experimental quantum optics, namely whether certain complex entangled quantum states can be constructed. The two main ideas behind our technique are:\\
1) To develop an edge pruning algorithm which helps to construct quantum optical experiments with as few resources as possible. \\
2) Detect a special sparse subgraph in the pruned graph corresponding to the quantum optical experiment whose edge count bounds the dimension of some multi-particle entangled quantum states.

Our ideas are general and might be useful to understand the experimental designs to construct several other quantum states like NOON states, cluster states, W states and Dickes states. With more structural insights into the graphs used for creating high-dimensional multi-particle entanglement, we believe our techniques can be used to resolve a conjecture on the constructability of certain complex entangled quantum states by Cervera-Lierta, Krenn, and Aspuru-Guzik \cite{CerveraLierta2022designofquantum}. This would give us more insights into quantum resource theory and the limitation of specific quantum photonic systems.

\section{Graph representation of quantum optics experiments}
\label{sec:intro}

Quantum entanglement theory implies that two particles can influence each other, even though they are separated over large distances. In 1964, Bell demonstrated that quantum mechanics conflicts with our classical understanding of the world \cite{bell}. Later, in 1989, Greenberger, Horne, and Zeilinger (abbreviated as GHZ) studied what can happen if more than two particles are entangled \cite{Greenberger}. Such states in which three particles are entangled (
$|GHZ_{3,2}\rangle = \frac{1}{\sqrt{2}}\left(|000\rangle + |111\rangle \right)$) were observed rejecting local-realistic theories  ~\cite{PhysRevLett.82.1345,Pan2000}. 

While the study of such states started purely out of fundamental curiosity \cite{fund_cur1,fund_cur2,fund_cur3}, they are now used in many applications in quantum information theory, such as quantum computing \cite{Gu2020}. They are also essential for early tests of quantum computing tasks \cite{quant_comp_tasks}, and quantum cryptography in quantum networks\cite{quant_networks}. Increasing the number of particles involved and the dimension of the GHZ state is essential both for foundational studies and practical applications. Motivated by this, a huge effort is being made by several experimental groups around the world to push the size of GHZ states. Photonic technology is one of the key technologies used to achieve this goal \cite{quant_comp_tasks,10photon}.

In 2017, Krenn, Gu and Zeilinger \cite{Quantum_graphs} discovered (and later extended \cite{Quantum_graphs_2, Quantum_graphs_3}) a bridge between experimental quantum optics and graph theory. They observed that large classes of quantum optics experiments (including those containing probabilistic photon pair sources, deterministic photon sources and linear optics elements) can be represented as an edge-coloured, edge-weighted graph. Additionally, every edge-coloured edge-weighted graph can be translated into a concrete experimental setup. This technique has led to the discovery of new quantum interference effects and connections to quantum computing \cite{Quantum_graphs_3}. Furthermore, it has been used as the representation of efficient AI-based design methods for new quantum experiments \cite{AIquantum1, AIquantum2}. The states formed by this framework were also experimentally demonstrated \cite{Feng:23, qian2023multiphoton}. This graph-based representation was also used to demonstrate many more systems beyond post-selected states (like NOON states and heralded states) \cite{AIquantum2}. This representation and another closely related graph-based representation have also been used for quantum circuit representation and computation \cite{AIquantum1, anand2022information}. The largest integrated photonic chip experiment (with several applications) presented so far \cite{Bao23} also follows a graph-based representation!


However, despite several efforts, a way to generate a GHZ state of dimension $d>2$ with more than $n=4$ photons with perfect quality and finite count rates without additional resources \cite{krenn2019questions} could not be found. This led Krenn and Gu to conjecture that it is not possible to achieve this physically. They have also formulated this question purely in graph theoretic terms and publicised it widely among graph theorists for a resolution \cite{mixon_website}. We now formally state this problem in graph-theoretic terms and explain its equivalence in quantum photonic terms. For a detailed description of why both problems are equivalent, we refer the reader to \cite{krenn2019questions}.

\subsection{Mathematical formulation of the experiments}

\begin{figure}[t!]
    \centering   
    \begin{minipage}{0.85\textwidth}
\centering    
{\includegraphics[width=63mm]{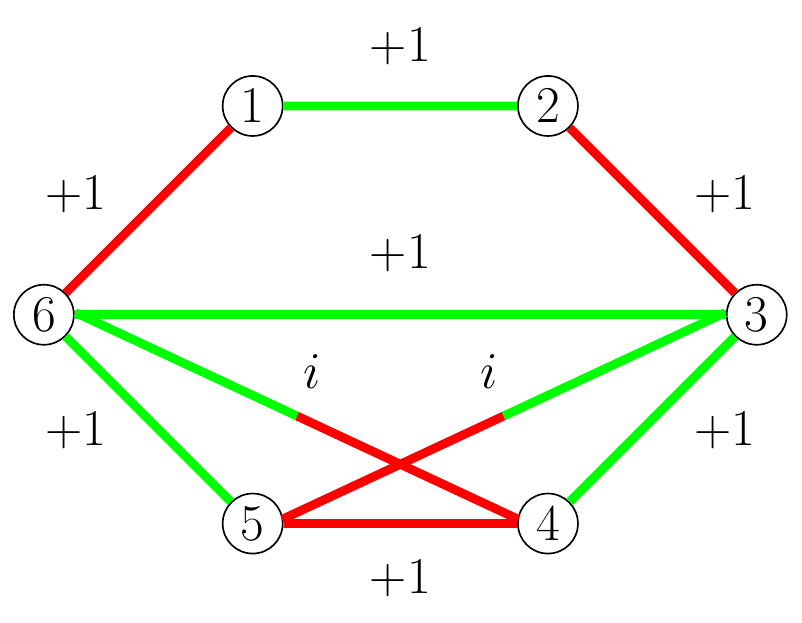}}
\caption{An edge-coloured edge-weighted graph. The edges $\{1,2\},\{3,4\},\{5,6\}, \{3,6\}$ are of colour green (mode number $1$) and $\{1,6\},\{2,3\},\{4,5\}$ are of colour red (mode number $0$). The edges $\{4,6\}$ and $\{3,5\}$ are bi-chromatic, where the halves starting at the vertices $4,5$ are of the colour red, and the remaining halves are of the colour green.}
\label{fig:main_example}
    \end{minipage}
   \hspace*{1cm}
\end{figure}

\begin{figure}[t!]
    \centering
\centering    
\begin{subfigure}[b]{0.2\textwidth}
         \centering
         \fbox{\includegraphics[width=0.9\textwidth]{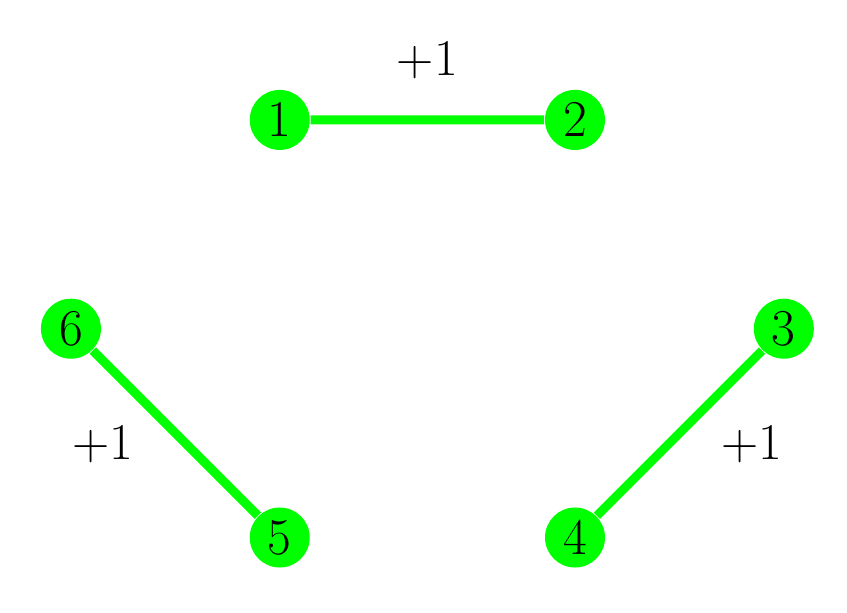}}
         \caption{The  vertex colouring is  $|111111\rangle$.}
         \label{fig:pm1}
\end{subfigure}
\hfill
\begin{subfigure}[b]{0.2\textwidth}
         \centering
         \fbox{\includegraphics[width=0.9\textwidth]{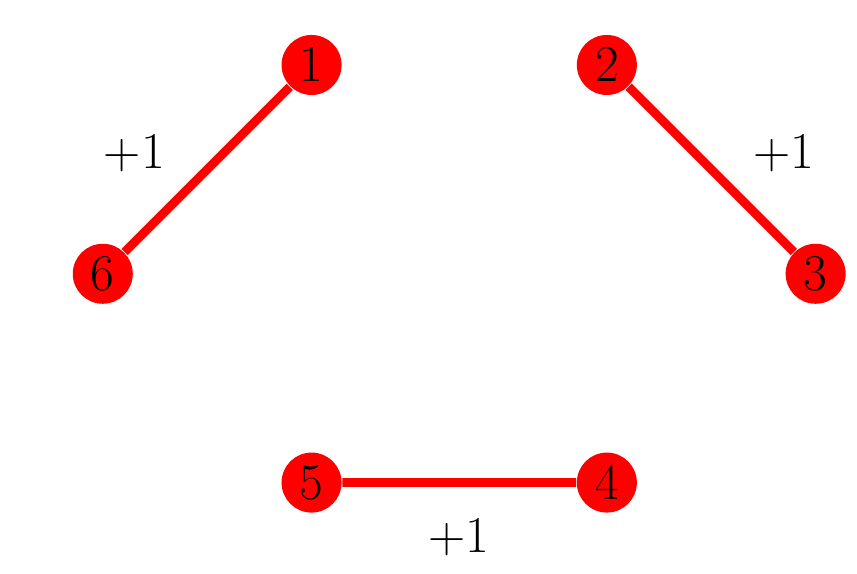}}
         \caption{The  vertex colouring is $|000000\rangle$.}
         \label{fig:pm2}
\end{subfigure}
\hfill
\begin{subfigure}[b]{0.2\textwidth}
         \centering
         \fbox{\includegraphics[width=0.9\textwidth]{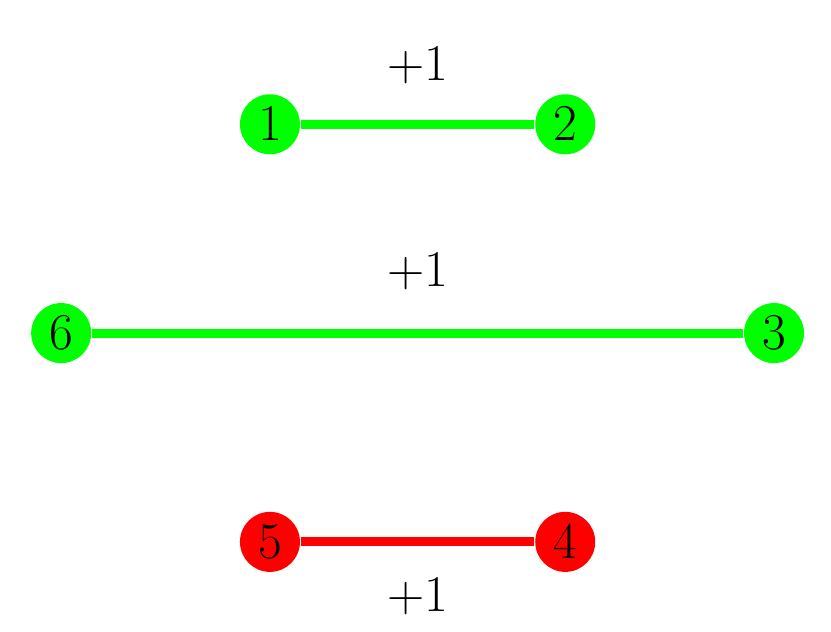}}
         \caption{The  vertex colouring is  $|111001\rangle$.}
         \label{fig:pm3}
\end{subfigure}
\hfill
\begin{subfigure}[b]{0.2\textwidth}
         \centering
         \fbox{\includegraphics[width=0.9\textwidth]{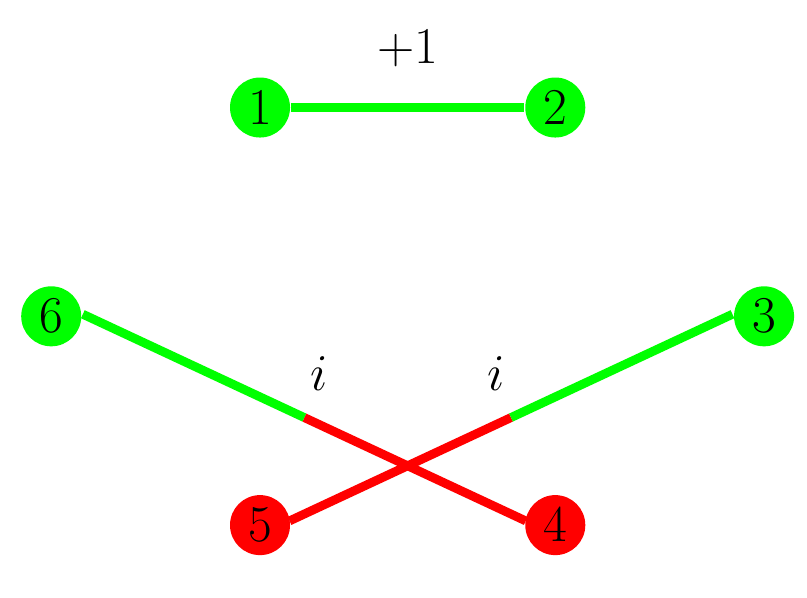}}
         \caption{The  vertex colouring is $|111001\rangle$.}
         \label{fig:pm4}
\end{subfigure}

\label{fig:perfect_matchings}
\end{figure}
We first define some commonly used graph-theoretic terms. For a graph $G$, let $V(G), E(G)$ denote the set of vertices and edges, respectively. For $S\subseteq V(G)$, $G[S]$ denotes the induced subgraph of $G$ on $S$. $\mathbb{N},\mathbb{C}$ denote the set of natural and complex numbers, respectively. The cardinality of a set $\cal{S}$ is denoted by $|\cal{S}|$. For a positive integer $r$, $[r]$ denotes the set $\{1,2\ldots,r\}$. $K_n$ denotes the complete graph on $n$ vertices. We call a subset $M$ of edges in a graph a matching if each vertex in the graph has at most one edge in $M$ incident on it. We call a matching $M$ to be a perfect matching if each vertex in the graph has exactly one edge in $M$ incident on it. For example, all four perfect matchings of the graph in 
\Cref{fig:main_example} are listed in 
\Cref{fig:perfect_matchings}.

We assume all graphs to be simple throughout unless otherwise mentioned. Usually, in an edge colouring, each edge is associated with a natural number. However, in such edge colourings, the edges are assumed to be monochromatic. But in the graphs corresponding to experiments, we are allowed to have bichromatic edges, i.e. one half coloured by a certain colour and the other half coloured by a different colour. For example, in the graph shown in \Cref{fig:main_example}, the edge between vertices $4$ and $6$ is a bichromatic edge. We now develop some new notation to describe bichromatic edges. 

Each edge of the graph can be thought to be formed by two half-edges, i.e., the edge $e=\{u,v\}$ consists of  the half-edge starting from the vertex $u$
to the middle of the edge $e$ (hereafter referred to as the $u$-half-edge of $e$) and the half-edge starting from the vertex $v$ to the
middle of the edge $e$ (hereafter referred to as the $v$-half-edge $e$). Thus, the edge set $E$ of the graph gives rise to the set of half-edges $H$,
with $|H| = 2 |E|$.  If $e =\{u,v\}$, we may denote the $v$-half-edge of $e$ by $e_v$ and $u$-half edge of $e$ by $e_u$. Consider the edge $e$ between vertices $4$ and $6$ in \Cref{fig:main_example}. The $4$-half edge of $e$ ($e_4$) is of colour $0$ (red), and the $6$-half edge ($e_6$) is of colour $1$ (green).

The type of edge colouring that we consider in this paper is more aptly called a {\it half-edge colouring}. It is a function from  $H$ to  $\mathbf {N}$, say $c: H \rightarrow \mathbf {N}$.  (Note that we use 
positive integers to name the colours.) In other words, each half-edge gets a colour.   An edge is called monochromatic if both
its half-edges get the same colour (in which case we may use
$c(e)$ to denote this colour); otherwise, it is called a bi-chromatic edge. The colour-degree of a vertex $v$ with respect to a colour $i$, $deg(v, i)$ is the
number of half-edges with colour $i$ that is incident on $v$. For example, consider the edge $e$ between vertices $4$ and $6$ in \Cref{fig:main_example}. It is easy to see that $c(e_4)=0$ and $c(e_6)=1$. Consider the edge $e'$ between vertices $1$ and $6$. As $c(e'_1)=c(e'_6)=0$, $e'$ is monochromatic and moreover, $c(e')=0$. The colour degree of vertex $3$ with respect to colour $1$, $deg(3,1)=3$.




A weight assignment $w$ assigns every edge $e$ a weight $w(e) \in \mathbb{C} \setminus \{0\}$. Note that, as a zero-weight edge is the same as the edge being absent, we don't need to consider zero-weight edges.

\begin{definition}
A half-edge coloured, edge-weighted graph is said to be an experiment graph.
\end{definition}



\begin{definition}
The weight of a perfect matching $P$, $w(P)$ is the product of the weights of all its edges, i.e., $\prod\limits_{e\in P}w(e)$
\end{definition}

A vertex colouring $vc$ associates a colour $i$ to each vertex in the graph for some $i\in \mathbb{N}$.
We use $vc(v)$ to denote the colour of vertex $v$ in the vertex colouring $vc$. Consider a perfect matching $P$. We say that each perfect matching $P$ \textit{induces} a vertex colouring $vc$ if for each vertex $v$, $vc(v)$ is equal to $c(e_v)$ (i.e., the colour of $v$-half of $e$) where $e$ is the unique edge in $P$ incident on $v$. Note that this way, we give exactly one colour to each vertex. So, it is indeed a vertex colouring. We say that this vertex colouring $vc$ is induced by $P$. Note that different perfect matchings can induce the same vertex colouring. For instance, \Cref{fig:pm3} and \Cref{fig:pm4} induce the same vertex colouring ($|111001\rangle$) in which vertices $1,2,3,6$ are coloured green and vertices $4,5$ are coloured red. A vertex colouring $vc$ is defined to be feasible if there exists at least one perfect matching $P$, which induces $vc$.

\begin{definition}
The weight of a vertex colouring $vc$, $w(vc)$ is the sum of the weights of all perfect matchings $P$ inducing $vc$. 
\end{definition}
The weight of a vertex colouring, which is not feasible, is zero by default.

\begin{definition}\label{ghz_def}
An experiment graph is said to be a \textit{GHZ} experiment graph if:
\begin{enumerate}
    \item All feasible monochromatic vertex colourings have a weight of $1$.
    \item All non-monochromatic vertex colourings have a weight of $0$.
\end{enumerate}
\end{definition}

For example, consider the experiment graph shown in \Cref{fig:main_example}.
\begin{enumerate}
    \item Let $vc_1$ correspond to the vertex colouring in which all vertices are coloured $1$, as shown in \Cref{fig:pm1} and let $vc_2$ correspond to the vertex colouring in which all vertices are coloured $0$, as shown in \Cref{fig:pm2}. $$w(vc_1)=1 \And w(vc_2)=1$$
    \item Let $vc_3$ correspond to the vertex colouring in which vertices $\{1,2,3,6\}$ are coloured $1$ and vertices $\{4,5\}$ are coloured $1$, as shown in \Cref{fig:pm3} and \Cref{fig:pm4}. $$w(vc_3)=1-1=0$$
\end{enumerate}
Observe that all feasible monochromatic vertex colourings, i.e., $vc_1$ and $vc_2$, have weight $1$. All non-monochromatic vertex colourings, i.e., $vc_3$, $vc_4$ and several other non-monochromatic vertex colourings which are not feasible, have a weight of $0$. Thus, the experiment graph shown in \Cref{fig:main_example} is a GHZ experiment graph.

\begin{definition}
The dimension of a GHZ experiment graph $G$, $\mu(G)$ is the number of feasible monochromatic vertex colourings having a weight of $1$.
\end{definition}

\subsection{Correspondence to quantum photonic experiments}

In \Cref{tab:graph-optical-mapping}, we list what different graph theoretic terms correspond to in the quantum photonic experiments. At a high level, every experiment graph (half-edge-coloured edge-weighted graph) corresponds to a quantum photonic experiment. For this quantum photonic experiment to give rise to a GHZ state, the experiment designer should carefully tune the edge weights and edge colours. If the edge weights (amplitude of photon pairs) and edge colours (mode numbers of the corresponding photon pairs) are tuned in such a way that it leads to a GHZ experiment graph $G$ (i.e., satisfies the conditions in \Cref{ghz_def}), then the corresponding quantum photonic experiment would lead to $n$-particle $d$-dimensional GHZ state ($|GHZ_{n,d}\rangle$), where $n=V(G)$ and $d=\mu(G)$. Note that by tuning weights differently so that they satisfy a different set of conditions (analogous to \Cref{ghz_def}), quantum states like NOON states, cluster states, W states and Dickes states can also be generated.

\begin{table}[htbp]
    \centering
    \caption{Mapping between Graphs and Quantum Optical Experiments}
    \label{tab:graph-optical-mapping}
    \begin{tabular}{|p{5cm}|p{8cm}|}    
        \hline
        \textbf{Graph theory} & \textbf{Quantum photonics} \\
        \hline
        Vertex $v$ & A single-photon detector $det(v)$ in the output of some photon path \\
        \hline
        Edge, $e=\{u,v\}$ & A correlated photon pair $\{p(e_u),p(e_v)\}$ 
        
        The photons $p(e_u)$ and $p(e_v)$ are on photon paths leading to the photon detectors $det(u)$ and $det(v)$, respectively.
        \\
        \hline
        Edge Weight of $e=\{u,v\}$ & Amplitude of the photon pair $\{p(e_u),p(e_v)\}$ \\
        \hline
        Bichromatic edge $e=\{u,v\}$ &
        Mode numbers of the photons $p(e_u)$ and $p(e_v)$ are different \\
        \hline Monochromatic edge $e=\{u,v\}$ &  Mode numbers of the photons $p(e_u)$ and $p(e_v)$ are same\\
        \hline Perfect matching $P=\{\leftindex_1{e},\leftindex_2 {e}\cdots\}$, where the edge $\leftindex_i {e}=\{u_i,v_i\}$  & A multi-photon event in which the photons $ p(\leftindex_{1} {e}_{u_1}),p(\leftindex_{1} {e}_{v_1}),p(\leftindex_{2} {e}_{u_2}),p(\leftindex_{2} {e}_{v_2})\cdots$ are detected.\\
        \hline Weight of a perfect matching $P$ & Product of amplitudes of photon pairs $ \{ p(\leftindex_{1} {e}_{u_1}),p(\leftindex_{1} {e}_{v_1})\},\{p(\leftindex_{2} {e}_{u_2}),p(\leftindex_{2} {e}_{v_2})\}\cdots$  \\
        \hline Vertex colouring $vc$ & A multi-photonic term.
        
        This is a mapping of each photon detector to the mode number of the photon detected by it.  \\
        \hline Vertex colouring induced by a perfect matching $P$ & Every multi-photon event causes a multi-photonic term. Such a term is produced by mapping each photon detector $det(v)$ to the mode number of the photon $p(e_v)$, where $e$ is the unique edge in $P$ incident on $v$.
        
        There could be several multi-photon events which cause the same multi-photonic term.\\
        
        \hline Weight of a vertex colouring $vc$ &  Amplitude of the multi-photonic term corresponding to $vc$. This is the sum of weights of all events causing the multi-photonic terms inducing to $vc$. \\
        \hline Monochromatic vertex colouring & Multiphotonic terms in which all photon detectors are mapped to the same mode number.  \\
        \hline GHZ experiment graph of dimension $d$ 
        
        & $d$-dimensional GHZ state (maximally entangled state)
        
        There are $d$ Multiphotonic terms of amplitude $1$ in which all photon detectors are mapped to the same mode number. All other multiphotonic terms have weight zero.\\
        \hline
    \end{tabular}
\end{table}

{The experiment corresponding to the experiment graph shown in \Cref{fig:main_example} produces the terms (vertex colourings) $|000000\rangle$ and $|111111\rangle$ with amplitude (weight) $1$. All the remaining terms (vertex colourings) have amplitude (weight) $0$. Therefore, the quantum state observed in the experiment (after normalizing the terms) is $|GHZ_{6,2}\rangle = \frac{1}{\sqrt{2}}\left(|000000\rangle + |111111\rangle \right)$ (a $6$ vertex GHZ experiment graph of dimension $2$). For a detailed description of the quantum physical meaning of this setup, refer to \cite{krenn2019questions}.}

\subsection{Progress on the problem}

\begin{table}[htbp]
    \centering
    \caption{Current state-of-the-art results and conjectures}
    \label{tab:comparison}
    \begin{tabular}{|c|c|c|c|c|}
        \hline
       \multirow{2}{*}{\textbf{Reference}} & \textbf{Multi edges} & \multirow{2}{*}{\parbox{3cm}{\textbf{Bichromatic} \\ \textbf{edges}}}  & \textbf{Dimension} & \textbf{Status}\\ & & & & \\
        \hline
       \cite{krenn_website} & Allowed & Allowed & $n>4$, $d\leq2$ & { \Cref{krenn_gu_conj}}\\
        \hline
        \multirow{2}{*}{
        \cite{CerveraLierta2022designofquantum}} & \multirow{2}{*} {Allowed} & \multirow{2}{*} {Not allowed} & \multirow{2}{*}{\parbox{3cm}{$n=6$, $d < 3$ \\ $n=8$, $d < 4$}} & \multirow{2}{*}{\parbox{3cm}{Proved with \\ SAT solvers}} \\
        & & & & \\
        \hline
        \multirow{2}{*}{
        \cite{CerveraLierta2022designofquantum}} & \multirow{2}{*} {Allowed} & \multirow{2}{*} {Not allowed} & \multirow{2}{*} {$n>4$, $d < \dfrac{n}{2}$} & \multirow{2}{*} {\Cref{quantumconjecture}}\\
        & & & & \\
        \hline
        {
        \cite{Kevin}} & \multirow{2}{*} {Allowed} & \multirow{2}{*} {Allowed} &  {$n=4$, $d \leq 3$} & \multirow{2}{*}{\parbox{3cm}{Proved with \\ Gröbner basis}}  \\
        & & & & \\
        \hline
         \multirow{2}{*}{This work} & \multirow{2}{*} {Not allowed} &\multirow{2}{*} { Allowed} & \multirow{2}{*} {$n>4$, $d < \dfrac{n}{\sqrt{2}}$} & \multirow{2}{*}{\parbox{3cm}{Proved in \\ \Cref{main_thm}}} \\
        & & & & \\
        \hline
    \end{tabular}
\end{table}



Krenn and Gu conjectured that physically it is not possible to generate a GHZ state of dimension $d>2$ with more than $n=4$ photons with perfect quality and finite count rates without additional resources \cite{krenn2019questions}. In graph-theoretic terms, their conjecture states that
\begin{conjecture}[Krenn-Gu Conjecture]\label{krenn_gu_conj}
For a graph $G$ which is not isomorphic to $K_4$, $\mu(G)\leq 2$.   
\end{conjecture}
While proving their conjecture would immediately lead to new insights into resource theory in quantum optics, finding a counter-example would uncover new peculiar quantum interference effects of a multi-photonic quantum system. 

However, when multi-edges and bi-chromatic edges are allowed, even for a GHZ experiment graph with just $4$ vertices, the question of whether the dimension is bounded or not looks surprisingly challenging to prove analytically. Only recently, Mantey proved this with extensive use of computers \cite{Kevin}. So, a general bound on the dimension as the function of the number of vertices of the experiment graph remains elusive. This motivated researchers to look at the problem by restricting the presence of bi-chromatic edges and multi-edges. The state-of-the-art results and conjectures are listed in \Cref{tab:comparison}.

\textbf{Absence of destructive interference.} In quantum physical terms, this is equivalent to saying every $n$-photon amplitude which is generated can be observed in the detectors, i.e.,  there are no two multi-photon terms with the same modes but opposite amplitudes that cancel. This special case does not allow true quantum resources and reduces the problem to a simple combinatorial question on unweighted coloured graphs.

Bogdanov \cite{bogdanov} proved that the dimension of a $4$ vertex experiment graph is at most $3$, and an $n\geq 6$ vertex graph is at most $2$. Chandran and Gajjala \cite{Gajjala} gave a structural classification of experiment graphs of dimension $2$. They also proved that if the maximum dimension achievable on a graph without destructive interference is not one, the maximum dimension achievable remains the same even if destructive interference is allowed. Vardi and Zhang \cite{vardi1,vardi2} proposed new colouring problems which are inspired by these experiments and investigated their computational complexity.

\textbf{Absence of multi-edges.} In quantum physical terms, this is equivalent to saying that the two-photon correlations that built up the experimental result can only have one specific mode combination seen in their computational basis. For example, they can be $|0,0\rangle,|0,1\rangle,|1,0\rangle $ or $|1,1\rangle$ but not $|0,D\rangle$ with $|D\rangle = \dfrac{1}{\sqrt{2}} (|0\rangle + |1\rangle)$

Chandran and Gajjala \cite{Gajjala} proved that the dimension of an $n$ vertex experiment graph is at most $n-3$ even when bi-chromatic edges are present for simple graphs. Their techniques can be extended to get the same bound when multi-edges are present and bi-chromatic edges are absent.

\textbf{Absence of bi-chromatic edges.} In quantum physical terms, this is equivalent to saying that the two-photon correlations that built up the experimental result can only have the same mode in their computational basis. For example, they can be $|0,0\rangle,|1,1\rangle$ or $|0,0\rangle+|1,1\rangle$ but not $|0,1\rangle$ or $|1,0\rangle$.

Ravsky \cite{ravsky} proposed a claim connecting this problem to rainbow matchings, i.e., a matching in which all edges have a different colour. Using a result of Kostochka and Yancey \cite{Kostochka2012LargeRM}, he showed that the dimension of an $n$ vertex experiment graph is at most $n-2$. Neugebauer \cite{neugebauer} used these ideas and did computational experiments. For GHZ experiment graphs with a small number of vertices, Cervera-Lierta, Krenn, and Aspuru-Guzik \cite{CerveraLierta2022designofquantum} translated this question into a boolean equation system and found that the system does not have a solution using SAT solvers. In particular, they show that for graphs with monochromatic edges, GHZ states with $n=6$, $d \geq 3$ and $n=8$, $d\geq 4$ cannot exist. The authors further conjectured the following more general claim
\begin{conjecture}[Cervera-Lierta et al. conjecture]
\label{quantumconjecture}
It is not possible to generate an $n>4$ vertex experiment graph with dimension  $d\geq \dfrac{n}{2}$.
\end{conjecture}

Note that \Cref{quantumconjecture} is much weaker than \Cref{krenn_gu_conj}. One may wonder why should \Cref{quantumconjecture} be stated since its correctness would be implied if \Cref{krenn_gu_conj} is proved. However, \Cref{krenn_gu_conj} seems to be strong from the currently available evidence. In particular, the major part of the intuition to formulate \Cref{krenn_gu_conj} stems from the failure of computational experiments using PyTheus \cite{AIquantum2} to produce small graphs with desired properties. However, it is well known that many conjectures formulated based on experiments on graphs of small size can potentially be wrong. This is because many intricate structural features of graphs manifest themselves only in large graphs. In our opinion, experiment graphs with reasonably high (a constant $\geq 3$) dimensions may exist, albeit such graphs are likely to be extremely structured and, therefore, rare to find.




\subsection{Our result}

In a \emph{simple} GHZ experiment graph with $n$ vertices, as every vertex has at most $n-1$ neighbours and each colour has at least one monochromatic edge incident on each vertex (proved formally in \Cref{obs:monoedge}), a trivial upper bound of $n-1$ can be obtained on the number of colours and therefore, the dimension. We can achieve the same bound with a more careful argument even when multi-edges are allowed, but bi-chromatic edges are absent \cite{neugebauer}. All the results which appeared so far in the literature either improve an additive factor over this trivial bound or work only on graphs with at most $8$ vertices. In this work, using the technique of \textit{local sparsification} for experiment graphs, we overcome this barrier by proving the following:
\begin{theorem}
\label{main_thm}
It is not possible to generate a simple $n>4$ vertex GHZ experiment graph with dimension  $d\geq \dfrac{n}{\sqrt{2}}$.
\end{theorem}

In graph theoretic terms, \Cref{main_thm} is equivalent to stating that $\mu(G)\leq \dfrac{|V(G)|}{\sqrt{2}}$. This translates to saying that it is not possible to produce a GHZ state of $n$ particles, with $\geq \dfrac{n}{\sqrt{2}}$ dimensions using this graph approach without additional quantum resources (such as auxiliary photons). Note that our bound holds for simple graphs even when bi-chromatic edges are allowed.

\section{Edge pruning}
\label{prelim}
We introduce the concept of an edge minimum GHZ experiment graph. A GHZ experiment graph $G$ is said to be edge minimum if there is no other graph $G'$ such that $|V(G')|=|V(G)|$, $\mu(G)=\mu(G')$ and $|E(G')|< |E(G)|$. Such graphs correspond to the experiments to create GHZ states with a given number of particles and a given dimension using minimum resources and are of interest to experimental physicists. Note that all GHZ experiment graphs need not be edge minimal. For instance, \Cref{fig:main_example} is a GHZ experiment graph, but it is not edge minimum. This is because the graph in \Cref{fig:main_example_2} is also a GHZ experiment graph with $6$ vertices and dimension $2$ and has fewer edges than the graph in \Cref{fig:main_example}. On the other hand,  the graph in \Cref{fig:main_example_2} is an edge minimum GHZ experiment graph. This is because in any $6$ vertex GHZ experiment graph of dimension $2$, there must be at least two perfect matchings of different colours and, hence, at least $6$ edges. The reader may notice that proving \Cref{main_thm} for Edge minimum GHZ experiment graphs implies \Cref{main_thm} is true for all GHZ experiment graphs.

\begin{figure}[t!]
    \centering   
    \begin{minipage}{0.85\textwidth}
\centering    
{\includegraphics[width=63mm]{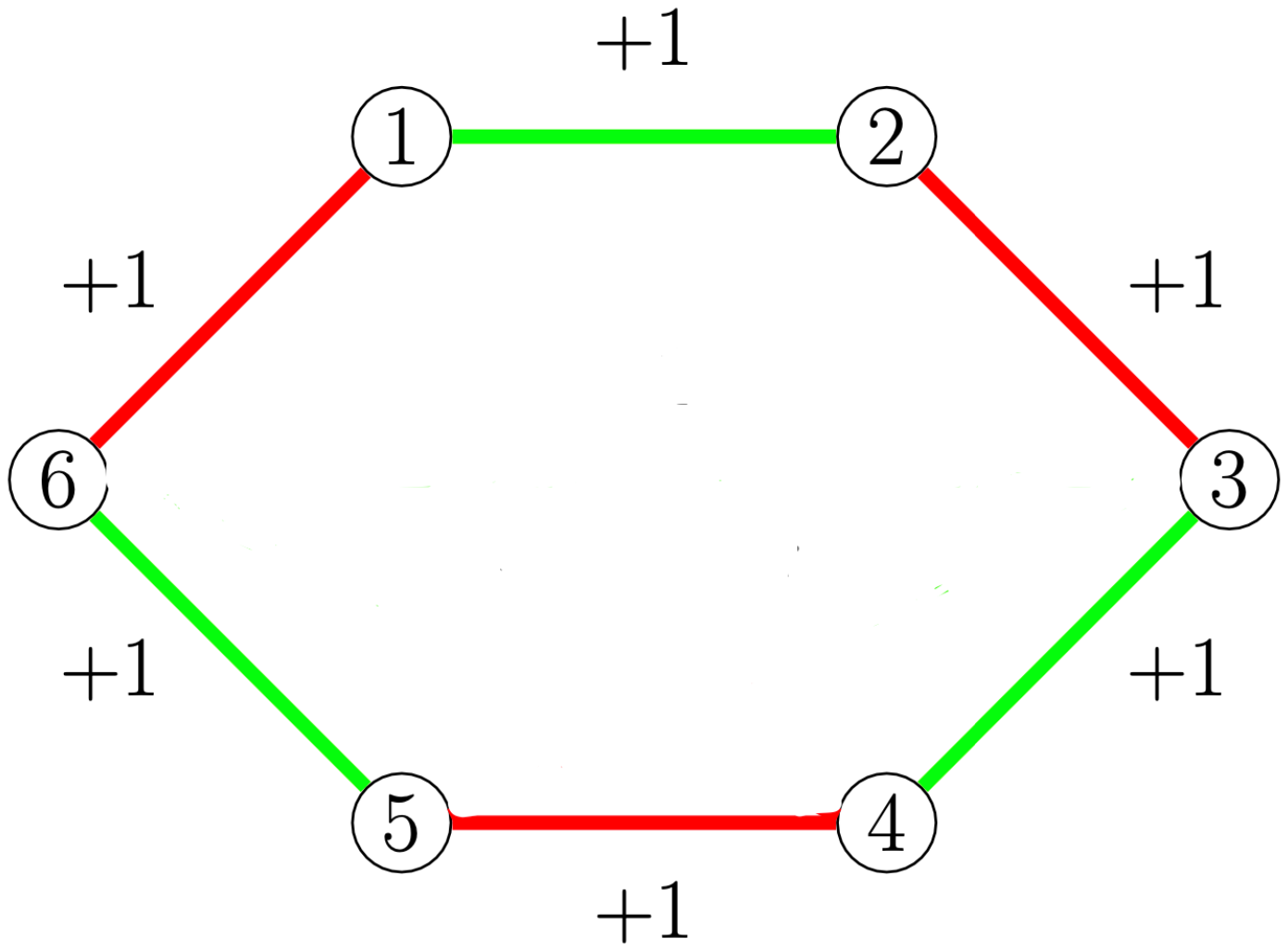}}
\caption{An edge minimum GHZ experiment graph}
\label{fig:main_example_2}
    \end{minipage}
   \hspace*{1cm}
\end{figure}


We define a colour $c$ to be \textit{participating} in an experiment graph $G$ if there is at least one half-edge in $G$ of colour $c$. We claim that in any edge minimum GHZ experiment graph, all monochromatic vertex colourings of participating colours are feasible. Towards a contradiction, let us assume that there exists an edge in the minimum GHZ experiment graph $G$ in which, for some $i$, there exists a monochromatic vertex colouring over $V(G)$ of colour $i$ which is not {feasible}. We can obtain a subgraph $G'$ by discarding all the bichromatic edges in which a half-edge is of colour $i$ and all monochromatic edges of colour $i$ from $G$. Observe that $G'$ would remain a GHZ experiment graph of the same dimension. This contradicts the edge minimality of $G$. Therefore, in an edge minimum GHZ experiment graph, all monochromatic vertex colourings of participating colours are feasible

A graph is matching covered if every edge of it is part of at least one perfect matching. If an edge $e$ is not part of any perfect matching $M$, then we call the edge $e$ to be redundant. By removing all redundant edges from the given graph $G$, we get its unique maximum matching covered sub-graph. For instance, the graph in \Cref{fig:mcg} is the matching covered graph of \Cref{fig:pruning_first}.
\begin{observation}
Edge minimum GHZ experiment graphs are matching covered.
\end{observation}
\begin{proof}
Suppose that there is an edge minimum GHZ experiment graph $G$ which is not matching covered. Consider its matching covered subgraph $H$. By definition, every perfect matching of $G$ is contained in $H$. Therefore, for any vertex colouring $vc$, the set of perfect matchings inducing it in $G$ and $H$ is the same, and therefore, the weight of the vertex colouring $vc$ in $H$ and $G$ are equal; hence, $H$ is also a GHZ experiment graph and $\mu(G)=\mu(H)$. Since $G$ is not matching covered, we know that $H$ is a proper subgraph of $G$ and hence $|E(H)|<|E(G)|$. This contradicts the edge minimality of $G$.
\end{proof}
While being matching covered is a necessary condition for a GHZ experiment graph, it is not a sufficient condition. For instance, \Cref{fig:main_example} is a GHZ experiment graph which is matching covered, but it is not edge minimum as we have discussed earlier.

\begin{observation}
\label{obs:monoedge}
In an edge minimum GHZ experiment graph $G$, for any vertex $v\in V(G)$ and colour $i \in [\mu(G)]$, there exists at least one monochromatic edge, say $e$ incident on $v$ such that $c(e)=i$.    
\end{observation}
\begin{proof}
By definition, the weight of the monochromatic vertex colouring with colour $i$ is non-zero. Therefore, there must be a monochromatic perfect matching $P$ in which all edges are of colour $i$. By the definition of perfect matching, for each vertex $v$, an edge $e \in P$ exists, which is incident on $v$. This $e$ must be of colour $i$ as $e \in P$.    
\end{proof}

{\definition{If an edge $e=\{u,v\}$ is monochromatic and $deg(u,c(e))=deg(v,c(e))=1$, then $e$ is said to be colour isolated.}}

In other words, if $e=\{u,v\}$ is colour isolated, it must be monochromatic, say of colour $i$, and there can not be any $u$-half edges other than $e_u$ and any other $v$-half edges other than $e_v$ which are of colour $i$. For instance, in  \Cref{fig:main_example}, the edges $\{1,6\}, \{1,2\}$ and $\{2,3\}$ are colour isolated.

\begin{lemma}
\label{c_isolated_edges}
Let $G$ be an edge minimum GHZ experiment graph. For an edge $e=\{u,v\}  \in E(G)$ incident on $u\in V(G)$ with $c(e_u)=i$. If $deg(u, i)=1$, then $e$ is monochromatic and $e$ is a colour-isolated edge. In other words, $deg(v, i)=1$.
\end{lemma}

\begin{proof}
Let $e=\{u,v\}$. Since $deg(u,i) =1$, there is exactly one half-edge of colour $i$ incident on  $u$, namely the $u$-half-edge of $e$.  But from  \Cref{obs:monoedge}, we know that 
there is at least one monochromatic edge of colour $i$ incident on $u$.
This implies that $e= \{u,v\}$ must be a monochromatic edge with $c(e)=i$. It follows that $deg(v, i) \geq 1$. Towards a contradiction, suppose $e$ is not a colour-isolated edge. Therefore, $deg(v, i)  \geq 2$. 

We now construct $G'$, a subgraph of $G$, by \emph{pruning} each edge $e'\neq e$,  which is incident on $v$ and has $c(e'_v)=c(e)=i$,
that is, those edges other than $e$ with their $v$-half-edges coloured $i$. We illustrate pruning in \Cref{fig:three_figures}.

Note that edge weights do not play any role in pruning. The pruned edges could be part of several perfect matchings, and it might be the case that the pruned edges are contributing to the weights of several vertex colourings. With the removal of such edges, the resultant graph might end up being an experiment graph, which is not a GHZ experiment graph. We show that this is not the case, i.e., we show that $G'$ is indeed a GHZ experiment graph! We do this by showing that the weight of any feasible vertex colouring $vc$ in $G'$ is equal to the weight of $vc$ in $G$ in \Cref{pruning_obs}. 

\begin{figure}
    \centering
    \begin{subfigure}[b]{0.32\textwidth}
        \centering
        \includegraphics[width=\textwidth]{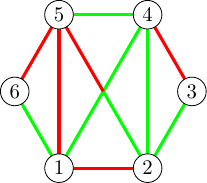}
        \caption{Graph $G$}
        \label{fig:pruning_first}
    \end{subfigure}
    \hfill
    \begin{subfigure}[b]{0.32\textwidth}
        \centering
        \includegraphics[width=\textwidth]{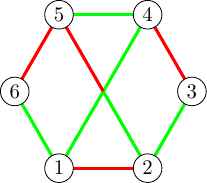}
        \caption{Matching covered graph of $G$}
        \label{fig:mcg}
    \end{subfigure}
    \hfill
    \begin{subfigure}[b]{0.32\textwidth}
        \centering
        \includegraphics[width=\textwidth]{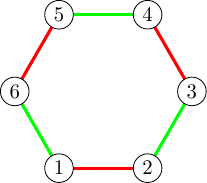}
        \caption{Pruned matching covered graph of $G$}
        \label{fig:pruning_second}
    \end{subfigure}
    \caption{The edges $\{1,5\}$ and $\{2,4\}$ are pruned as they are not part of any perfect matching. This gives us the matching covered graph of $G$ in \Cref{fig:mcg}. Observe that $deg(6,1)=1$. So we prune the edge $\{1,4\}$. Similarly, the edge $\{2,5\}$ is pruned as $deg(6,0)=1$. This leads to the pruned graph in \Cref{fig:pruning_second} ($G'$). 
    }
    \label{fig:three_figures}
\end{figure}

\begin{observation}\label{pruning_obs}
$G'$ (the subgraph of $G$ formed above by pruning) is a GHZ experiment graph.
\end{observation}
\begin{proof}
Let $vc$ be a feasible vertex colouring of $G'$. First, note that each perfect matching $P$ that can induce $vc$ in $G'$ is also present in $G$ since $G'$ is a subgraph of $G$. Therefore, $P$ will contribute to the weight of $vc$ in $G$ as well. We call a perfect matching $P'$ to be \textit{exclusive} if it is present in $G$ but not in $G'$. Our only concern is that there might be an exclusive perfect matching $P'$, which induces $vc$. We will show that this is not possible. We will prove this by showing that any exclusive perfect matching $P'$ of $G$ will always induce a vertex colouring which is infeasible in $G'$. Formally, consider an exclusive perfect matching $P$. Let $vc$ be the vertex colouring induced by $P$ in $G$. We now claim that $vc$ is not feasible in $G'$.

We first prove that $vc(v) = c(e) =i$ and $vc(u) \neq c(e)=i$. As $P$ is not contained in $G'$, there must be an edge $e'\in P$ such that $e'\in G$ and $e'\notin G'$. By construction, such an $e'\neq e$, must be incident on $v$ and has $c(e'_v)=i$. Therefore, $P$ induces the colour $c(e)=i$ on $v$ and hence $vc(v)=c(e)=i$. Since $e'$ is already incident on $v$, it must be the case that $e \notin P$, since $P$ is a perfect matching and by definition of perfect matching, two different edges incident on the same vertex cannot be simultaneously in $P$.
But $P$ should have an edge, say $e'' \ne e$ incident on $u$ also.  
As $deg(u,i)=1$ and $e \ne e''$, the color of the $u$-half-edge of $e''$
has to be different from $i$. It follows that $vc(u) =c(e''_u) \neq c(e)$. 

We now claim that any vertex colouring $vc$ with $vc(u) \neq c(e)$ and $vc(v) = c(e)$ is infeasible in $G'$. Towards a contradiction, suppose such a $vc$ is feasible. It follows that there must be a perfect matching $P'$ inducing the vertex colouring $vc$. By construction, $e$ is the only edge of colour $c(e)$ incident on $v$ in $G'$. Therefore, $e \in P'$. As $e$ is also incident on $u$, $P'$ induces the colour $c(e)$ on $u$. Also, it follows that $vc(u) = c(e)$. This contradicts the assumption that $vc(u) \neq c(e)$. Therefore, a vertex colouring $vc$ with $vc(u) \neq c(e)$ and $vc(v) = c(e)$ is infeasible in $G'$.

Therefore, if a vertex colouring $vc$ is feasible in $G'$, then $vc$ is induced by the same set of perfect matchings in both $G$ and $G'$. It follows that the weight of any feasible vertex colouring $vc$ in $G'$ is equal to the weight of $vc$ in $G$; hence, $G'$ is a GHZ experiment graph. 

\end{proof}
We have proved that for every perfect matching $P$ which contains the pruned edge $e'$, the the vertex colouring $vc$ induced by $P$ has the property that $vc(v) = c(e) =i$ and $vc(u) \neq c(e)=i$. Therefore, any such $vc$ is not monochromatic. Therefore, the monochromatic vertex colouring of colour $i$ in $G'$ has a weight equal to the weight of monochromatic vertex colouring of colour $i$ in $G$, which is $1$ for all $i \in [\mu]$. It follows that $\mu(G)=\mu(G')$. 

We will now conclude the proof of \Cref{c_isolated_edges}. Since $deg(v, i)  \geq 2$, there exists at least one edge $e'\neq e$ with its $v$-half-edges coloured $i$. By construction such an $e'$ will be pruned. Therefore, $G'$ will have at least one edge less than $G$. But this contradicts the edge minimality of $G$ as $G'$ is a GHZ experiment graph with $n$ vertices, $\mu(G)=\mu(G')$ and has fewer edges than $G$.

\end{proof}




\section{Bounding the dimension}
\subsection{Proof sketch}

If $\mu(G) \le n/2$, our theorem is trivially true since $n/2 \le n/\sqrt 2$.  So, we can concentrate on the case
where $\mu (G) > n/2$. In other words, technically speaking, we will be showing that if $\mu(G) > n/2$, the $\mu(G) \le n/\sqrt 2$, from
which  the stronger statement for any $G$, $\mu(G) \le n / \sqrt 2$, obviously follows. 
We will complete the proof in two steps.

\noindent {\bf Step 1:}  We will show that  there is a color $i \in [\mu(G)]$, such that there exists
a matching  (referred to as a special matching)   $M$  with $|M| \ge  n-\dfrac{n^2}{2\mu}$ such that all the edges of $M$ are color-isolated and are 
of color $i$.  The details  of this will be presented in section \Cref{sec:larg_match}.

\noindent {\bf Step 2:} In \Cref{structural_analysis}, we will  show that  if $M$ is a special matching
 then $\mu \leq n-|M|$.

Combining Step 1 and Step 2, we can then immediately infer that    
$$\mu \leq   n-|M| \leq n-(n-\dfrac{n^2}{2\mu}) \leq \dfrac{n^2}{2\mu}$$
Therefore, the required result  $\mu \leq \dfrac{n}{\sqrt{2}}$ would follow. 

Now we proceed to prove the claims in Step 1 and Step 2.

\subsection{Existence of a large special matching}
\label{sec:larg_match}
Let $k(v)$ denote the number of colours such that the colour degree of
$v$ with respect to that colour is exactly $1$.  In notation $k(v) = |\{ i \in [\mu]: 
deg(v,i) = 1\}|$. 
\begin{observation}
\label{deg_obs}
For all $v \in V(G)$, $k(v) \ge 2\mu-n+1$. 
\end{observation}
\begin{proof}
Suppose \Cref{deg_obs} is false. Then there exists a $v \in V(G)$ such that $k(v) < 2 \mu -n + 1$. Since $k(v)$ is integral, $k(v) \leq 2 \mu -n $.  We can now assume
that  $k(v) = 2\mu -n -i$ for some $i \in  [0,2\mu-n]$. Note that when $i=0$, $k(v)$ is $2\mu -n$ and when $i=2\mu-n$, $k(v)$ is $0$. $2\mu -n -i$ colors have colour-degree exactly $1$ on  $v$;  therefore, the remaining  $n-\mu+i$ colours have colour-degree  at least $2$ on $v$.
Counting one half-edge for each of the first types of colours and at least two half-edges for each of the second type of colours, we get that there are $2\mu-n-i+2(n-\mu+i)=n+i\geq n$ neighbours to $v$ (as $i \geq 0$). But since $G$ is a simple graph, $v$ can have at most $n-1$ neighbours. Contradiction.  Since we started with
the assumption $k(v) < 2 \mu -n + 1$, the converse should be true. 
\end{proof}

\begin{observation}
\label{large_col_iso_edges}
$G$ has at least $(2\mu-n+1)\dfrac{n}{2}$ colour isolated edges.
\end{observation}

\begin{proof}
From \Cref{c_isolated_edges}, we know that for every colour $i$ with $deg(u, i)=1$, there is a colour isolated edge of colour $i$ incident on $u$. Along with \Cref{deg_obs}, this implies that there are at least $2\mu-n+1$ colour-isolated edges that are incident on each vertex. Since there are $n$ vertices and each colour-isolated edge is incident on exactly two vertices, there must be at least $(2\mu-n+1)\dfrac{n}{2}$ colour-isolated edges in $G$.
\end{proof}

\begin{theorem}
\label{large_red_matching}
For some $i\in [\mu]$, there exist at least $\dfrac{(2\mu - n)n}{2\mu}$ $i$-coloured isolated edges.
\end{theorem}
\begin{proof}
From \Cref{large_col_iso_edges}, we know that there are at least  $(2\mu-n+1)\dfrac{n}{2}$  colour isolated edges. Since there are $\mu$ colours, by a simple averaging argument, we get that for some colour $i\in [\mu]$, there exist at least $\dfrac{(2\mu - n)n}{2\mu}$ $i$-coloured isolated edges.
\end{proof}

\subsection{Detecting a sparse subgraph}
\label{structural_analysis}

\begin{figure}[t!]
    \centering
\centering    
\begin{subfigure}[b]{0.42\textwidth}
         \centering
         \fbox{\includegraphics[width=75mm]{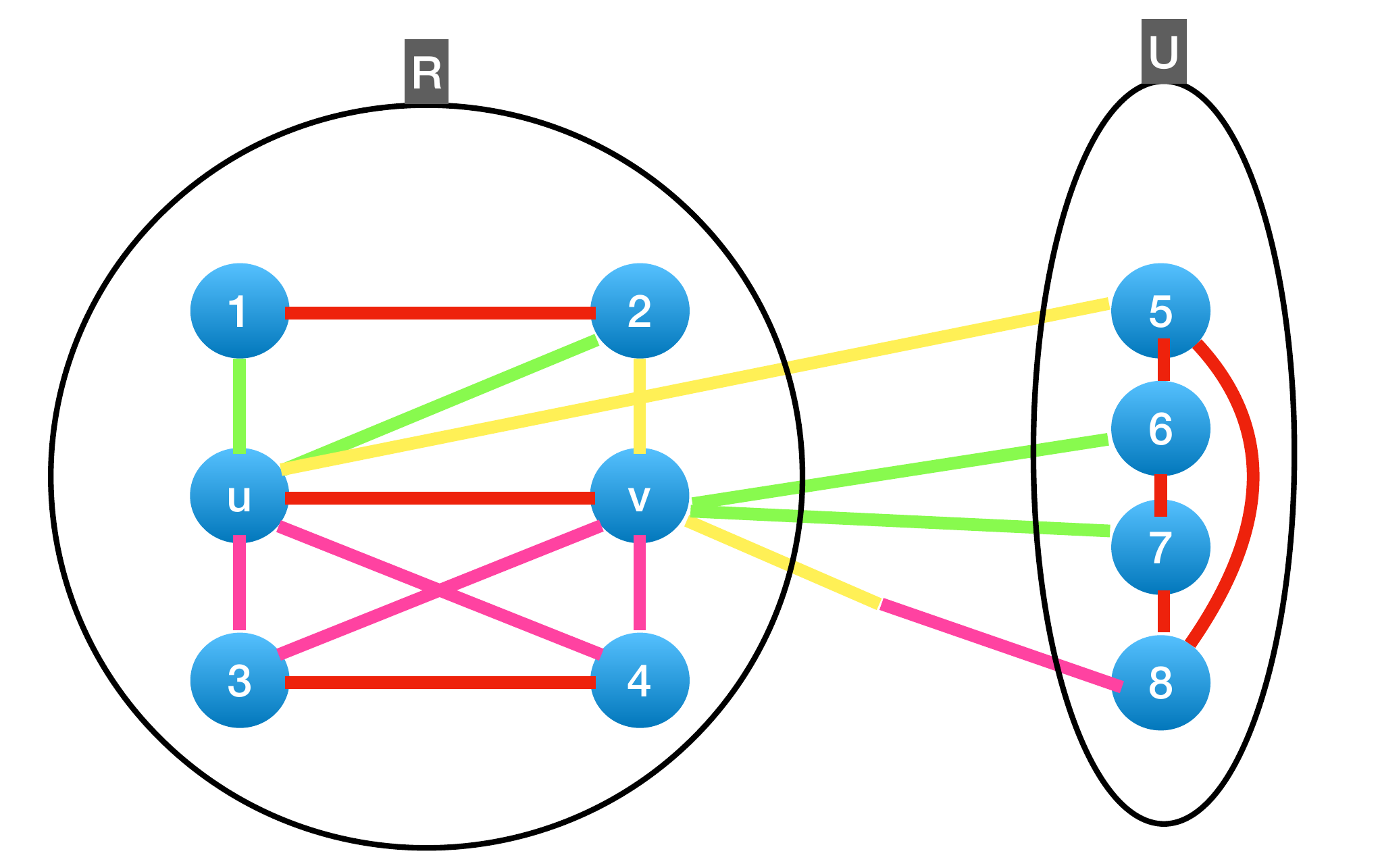}}
         \caption{Coloured graph $H$ where red is colour $1$}
\end{subfigure}
\hfill
\begin{subfigure}[b]{0.4\textwidth}
         \centering
         \fbox{\includegraphics[width=45mm]{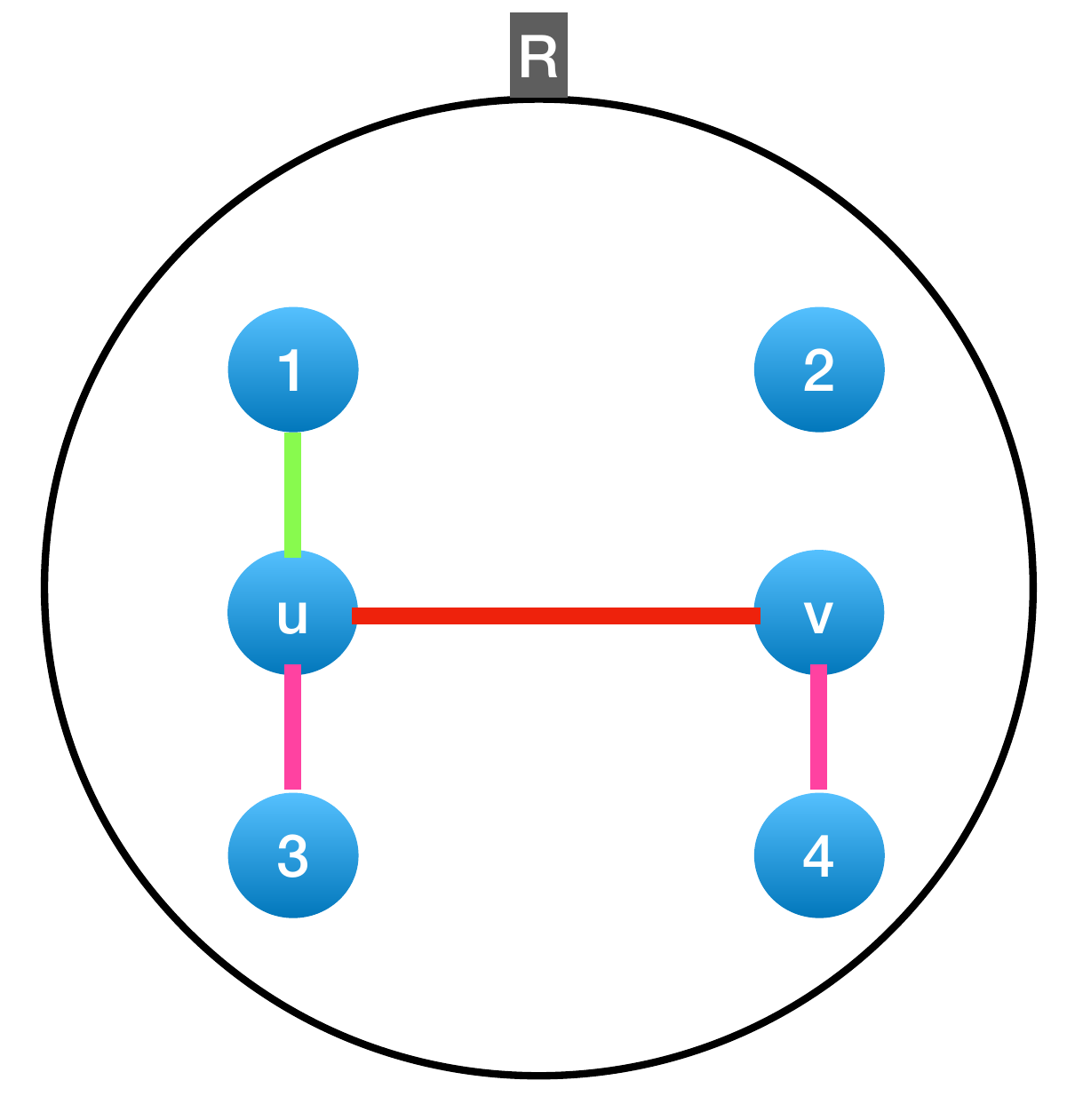}}
         \caption{Representative sparse graph $\chi(H)$} 
\end{subfigure}

\caption{Constructing a representative sparse graph:\\ $E(\chi_{u})=\{\{u,1\},\{u,3\},\{u,v\}\}$ and $E(\chi_{v})=\{\{v,4\},\{u,v\}\}$}
\label{fig:char_graph}
\end{figure}

Without loss of generality, we will assume that the large matching exhibited in \Cref{large_red_matching} has colour $1$ for the remainder of this section. Let $R$ represent the set of vertices whose colour degree with respect to the colour $1$ is one. Let $U$ be the set of vertices whose colour degree with respect to the colour $1$ is at least two. Clearly, $R \cup U= V(G)$.
\begin{observation}
$|R| \geq \dfrac{(2\mu-n)n}{\mu}$ and $|U| \leq \dfrac{n^2}{\mu}-n$.
\end{observation}
\begin{proof}
From \Cref{c_isolated_edges}, every vertex $v \in R$ has exactly one $1$-coloured isolated edge incident on it, and each $1$-coloured isolated edge is incident on two vertices in $R$. Therefore, there are exactly $\dfrac{|R|}{2}$ $1$-coloured isolated edges forming a perfect matching in the induced subgraph $G[R]$. Along with \Cref{large_red_matching}, this would imply that $|R| \geq \dfrac{(2\mu-n)n}{\mu}$ . Since $|R|+|U|=n$, we infer that $|U| \leq \dfrac{n^2}{\mu}-n$. 
\end{proof}

Note that since $n \geq 6$, we have $|R| \geq \dfrac{(2\mu-n)n}{\mu} = (2-\dfrac{n}{\mu})n \geq (2-\sqrt{2})n \geq 3.5$. Since $|R|$ must be integral, $|R|\geq 4$. Therefore, there exists at least two $1$-coloured isolated edges.

We now pick an arbitrary $1$-coloured isolated edge $\{u,v\}$. For the remainder of this section, we base our analysis on the edges incident on the vertices $u,v$. We now construct a special subgraph on $R$ vertices (which is a subgraph of the induced subgraph $G[R]$) called the representative sparse graph $\chi$ using some of the edges incident on $u,v$.

We first construct the graph $\chi_u$ in the following way over $R$: For every colour $i\in [\mu]$, if there is no vertex $w \in U$ such that the $u$-half edge of the edge $\{u,w\}$ has colour $i$ (i.e., $c(\{u,w\},u)=i$), then add an arbitrary monochromatic edge of colour $i$ (its existence is guaranteed from \Cref{obs:monoedge}) to $E(\chi_u)$. Note that this contains the edge $\{u,v\}$ also. 

Similarly, we construct the graph $\chi_v$ in the following way over $R$: For every colour $i\in [\mu]$, if there is no vertex $w \in U$ such that the $v$-half edge of the  edge $\{v,w\}$ has colour $i$ (i.e., $c(\{v,w\},v)=i$), then add an arbitrary monochromatic edge of colour $i$ (its existence is guaranteed from \Cref{obs:monoedge}) to $E(\chi_v)$.

We define the graph $\chi$ over $R$ with the edge set $E(\chi)=E(\chi_u) \cup E(\chi_v)$. An {example} graph and its representative sparse graph are presented in \Cref{fig:char_graph}.

\begin{lemma}
\label{lem:sumoftwo}
$2\mu \leq 2|U|+|E(\chi)|+1 $    
\end{lemma}
\begin{proof}
For any colour $i \in [\mu]$, there is either a vertex $w \in U$ such that $c(e_u)=i$, where $e=\{u,w\}$ or there is a monochromatic edge of colour $i$ in $\chi_u$. Therefore, $$\mu \leq |U|+|E(\chi_u)|$$ Similarly, $$\mu \leq |U|+|E(\chi_v)|$$
From the above two equations, we get that 
$$2\mu \leq 2|U|+|E(\chi_u)|+|E(\chi_v)|$$

It is easy to see that $E(\chi_u) \cap E(\chi_v)=\{\{u,v\}\}$ and hence $|E(\chi)|=|E(\chi_u)|+|E(\chi_v)|-1$. It follows that 
$$2\mu \leq 2|U|+|E(\chi_u)|+|E(\chi_v)|=2|U|+|E(\chi)|+1$$
\end{proof}

Since every edge in $\chi_u$ is incident on $u$ and ends in $R\setminus \{u\}$, we can infer that $|E(\chi_u)| \leq |R|-1$. Similarly, since every edge in $\chi_v$ is incident on $v$ and ends in $R\setminus \{v\}$, we can infer that $|E(\chi_v)| \leq |R|-1$. As $E(\chi_u) \cap E(\chi_v)=\{\{u,v\}\}$, it follows that $|E(\chi)| \leq 2|R|-3$. But this upper bound for $|E(\chi)|$ can be strengthened by using the following structural observation. Informally, we prove that configurations like \Cref{fig:e0,fig:e1,fig:e2} are possible and configurations like \Cref{fig:e3,fig:e4} are not possible. 
\begin{lemma}
\label{struct_lemma}
If $\{u',v'\}$ is a $1$-coloured isolated edge in $G[R]$ distinct from $\{u,v\}$, then $$|\{\{u,u'\},\{u,v'\},\{v,u'\},\{v,v'\}\} \bigcap E(\chi)|\leq2$$    
\end{lemma}
\begin{proof}
Suppose $|\{\{u,u'\},\{u,v'\},\{v,u'\},\{v,v'\}\} \bigcap \chi| \geq 3$. Without loss of generality, let the edges $\{u,u'\},\{u,v'\},\{v,v'\}$ be present in $E(\chi)$ (as all $4$ edges are symmetrically positioned). Recall that by construction, all edges in $E(\chi)$ are monochromatic. So, let the edges $\{u,u'\},\{u,v'\},\{v,v'\}$ be of colours $i,j,k$, respectively. Since $E(\chi_u)$ and $E(\chi_v)$ contain at most one edge of any colour, we know that $i\neq j$ and $i,j,k$ are not equal to $1$. Note that $k$ might be equal to $i$ or $j$.

\begin{figure}[t!]
    \centering
\centering    
\begin{subfigure}[b]{0.18\textwidth}
         \centering
         \fbox{\includegraphics[width=0.8\textwidth]{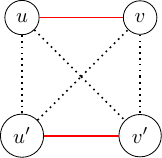}}
         \caption{Possible}
         \label{fig:e0}
\end{subfigure}
\hfill
\begin{subfigure}[b]{0.18\textwidth}
         \centering
         \fbox{\includegraphics[width=0.8\textwidth]{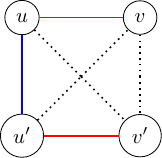}}
         \caption{Possible}
         \label{fig:e1}
\end{subfigure}
\hfill
\begin{subfigure}[b]{0.18\textwidth}
         \centering
         \fbox{\includegraphics[width=0.8\textwidth]{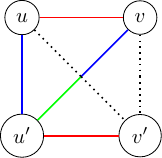}}
         \caption{Possible}
         \label{fig:e2}
\end{subfigure}
\hfill
\begin{subfigure}[b]{0.18\textwidth}
         \centering
         \fbox{\includegraphics[width=0.8\textwidth]{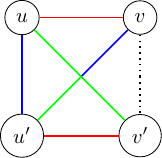}}
         \caption{Not possible}
         \label{fig:e3}
\end{subfigure}
\hfill
\begin{subfigure}[b]{0.18\textwidth}
         \centering
         \fbox{\includegraphics[width=0.8\textwidth]{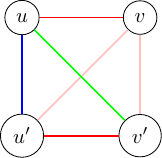}}
         \caption{Not possible}
         \label{fig:e4}
\end{subfigure}

\caption{Examples for \Cref{struct_lemma}. Dotted lines denote the edge being absent}
\label{fig:examples_size}
\end{figure}

Consider the vertex colouring $vc$ in which $u,u'$ get the colour $i$, $v,v'$ get the colour $k$, and the vertices in $V(G)\setminus \{u,u',v,v'\}$ are coloured $1$. Note that $V(G)\setminus \{u,u',v,v'\}$ is non-empty as $|V(G)|>4$. Let $vc'$ be the vertex colouring in which every vertex in $V(G)$ gets the colour $1$.
\begin{observation}
$\dfrac{w(vc)}{w(vc')}=\dfrac{w(\{u,u'\})w(\{v,v'\})}{w(\{u,v\})w(\{u',v'\})}$
\end{observation}
\begin{proof}
Let $M$ be the set of all $1$-coloured isolated edges. As every vertex in $R$ has exactly one $1$-coloured edge incident on it (and no $1$-coloured half-edges incident on it), every perfect matching $P'$ inducing $vc'$ must contain all edges in $M$. Let $W$ denote the weight of the monochromatic vertex colouring of colour $1$ on $G[U]$. As the vertices in $U$ match with themselves in every perfect matching $P'$ (since there is no $1$-coloured edge connecting $U$ to $R$), it is easy to see that
$$w(vc')= W\times \prod\limits_{e\in M} w(e)$$

As every vertex in $R\setminus\{u,v,u',v'\}$ has exactly one $1$-coloured edge incident on it, every perfect matching $P$ inducing $vc$ must contain all edges in $M\setminus \{\{u,v\},\{u',v'\}\}$. By definition of the representative sparse graph, the vertex $u$ can obtain the colour $i$ only through an edge $\{u,w\}$ such that $w \in R$. But the vertices $R\setminus \{u,v,u',v'\}$ are already matched and $c(\{u,v\}_u)=c(\{u,v\})=1\neq i$ and $c(\{u,v'\}_u)=c(\{u,v'\})=j\neq i$. Therefore, the edge $\{u,u'\}$ must be present in every perfect matching $P$ that induces $vc$. Again, by definition of the representative sparse graph, the vertex $v$ can obtain the colour $k$ only through an edge $\{v,w\}$ such that $w \in R$. But all the vertices in $R\setminus \{v,v'\}$ are already matched. Therefore, the edge $\{v,v'\}$ must be present in every perfect matching $P$. The remaining vertices in $U$ should match among themselves in every perfect matching $P$ that induces $vc$. It is now easy to see that the weight of the vertex colouring $vc$ is
$$w(vc)=
W \times w(\{u,u'\}) \times w(\{v,v'\}) \times \prod\limits_{e \in M\setminus\{\{u,v\},\{u',v'\}\}}w(e) $$

$$= \dfrac{w(\{u,u'\})w(\{v,v'\})}{w(\{u,v\})w(\{u',v'\})} \times W \times \prod\limits_{e \in M}w(e)$$
It follows that 
$$\dfrac{w(vc)}{w(vc')}=\dfrac{w(\{u,u'\})w(\{v,v'\})}{w(\{u,v\})w(\{u',v'\})}$$

\end{proof}

Since $vc'$ is a monochromatic vertex colouring, $w(vc')=1$. Recall that all edge weights are non-zero by definition. Therefore, $w(vc)$ is non-zero. But this is a contradiction as $vc$ is a non-monochromatic vertex colouring and should have weight $0$ by definition.

\end{proof}

\begin{lemma}
\label{lem:boundchi}
$|E(\chi)| \leq |R|-1$    
\end{lemma}
\begin{proof}
From \Cref{struct_lemma}, we know that $\chi$ can contain at most two edges between $\{u,v\}$ and any other $1$-coloured isolated edge. Therefore, the total number of edges in $\chi$ is $$E(\chi) \leq 2(\dfrac{|R|-2}{2})+1=|R|-1$$
\end{proof}

\begin{theorem}
$\mu \leq \dfrac{n}{\sqrt{2}}$    
\end{theorem}
\begin{proof}
From \Cref{lem:sumoftwo}, we have 
$$2\mu \leq   2|U|+|E(\chi)|+1 $$
Using \Cref{lem:boundchi}, we get
$$2\mu \leq   2|U|+|R| $$
As $|U|+|R|=n$ by definition,
$$2\mu \leq   2n-|R| $$
It follows from \Cref{large_red_matching} that $|R|$ is at least $\dfrac{(2\mu-n)n}{\mu}$ and hence
$$2\mu \leq 2n-\dfrac{(2\mu-n)n}{\mu} = \dfrac{n^2}{\mu}$$
Therefore, $\mu \leq \dfrac{n}{\sqrt{2}}$.
\end{proof}

\section{Conclusion}
Whether or not GHZ states of higher dimensions can be created in quantum photonic experiments is an important open question with little progress despite several efforts. Making use of the graph-theoretic formulation of this question, we present the first non-trivial upper bounds on the dimension of GHZ states which can be created. In particular, we prove that $n$ particle $d$-dimensional states can not be created with large dimension (i.e., with $d>\dfrac{n}{\sqrt{2}}$) using a large class of experiments. 

Our proof uses the framework of finding contradicting properties of the minimal counter-example. This framework is quite general (and commonly used in graph theory) and is extendable to any quantum state. However, the challenging part would be to understand which contradicting properties one should look for in the other quantum states. Some properties, like the minimal counterexample being a matching covered graph, are quite general and extend to the experiment graphs corresponding to any quantum state! On top of this, one may design other pruning strategies to suit the quantum state of interest and then find contradicting properties. The idea of colour-isolated edges could be helpful in finding such contradicting properties.

For instance, consider a W state $|W_n\rangle = \dfrac{1}{\sqrt{n}}  \hat{S}\left(|0\rangle^{\otimes(n-1)}|1\rangle\right)$ where $n$ is the number of particles and $\hat{S}$ is the symmetrical operator that gives summation over all distinct permutations of the $n$ particles. Using our pruning ideas, one can get useful properties for a minimal experiment graph, which can generate $|W_n\rangle$. For instance, every vertex $v$ would have exactly one half-edge of colour $1$ incident on it, and the remaining half-edges must be of colour $0$. This is consistent with the construction of W states described in \cite{Quantum_graphs_3}. Our techniques are not immediately extendable for systems beyond post-selected states. However, it is known that graph-based representation can lead to many more systems beyond post-selected states (like NOON states and heralded states) \cite{AIquantum2}. Extending our techniques for such states is an important open question we leave for future work.

\bibliographystyle{unsrtnat}
\bibliography{mybibliography.bib}

\end{document}